\documentclass{IEEEtran}
\usepackage{graphicx}
\usepackage{amssymb}
\usepackage{eucal}
\usepackage{epstopdf}
\usepackage{amsmath}
\usepackage{enumerate}
\usepackage{cite}
\usepackage{mathcomp}
\usepackage{supertabular}
\newcommand{\C}{{\mathcal C}}
\newcommand{\D}{\mathcal{D}}

\newcommand{\F}{{\mathcal F}}

\newcommand{\vM}{{\sf M}}
\newcommand{\vR}{{\sf R}}

\newcommand{\bbF}{\mathbb{F}}

\newcommand{\bbZ}{{\mathbb Z}}

\newtheorem{corollary}{Corollary}[section]

\newtheorem{example}{Example}[section]
\newtheorem{lemma}{Lemma}[section]
\newtheorem{proposition}{Proposition}[section]
\newtheorem{theorem}{Theorem}[section]

\title{Optimal Partitioned Cyclic Difference Packings for Frequency Hopping and Code
Synchronization}

\author{    Yeow~Meng~Chee,~\IEEEmembership{Senior Member,~IEEE},
		Alan~C.~H.~Ling,~and~Jianxing~Yin
		\thanks{The research of Y. M. Chee is supported in part by
		the National Research Foundation of Singapore under Research
		Grant NRF-CRP2-2007-03, and the
		Nanyang Technological University under Research Grant M58110040.
		The research of J. Yin is supported by NSFC Project Number 10671140.}%
		\thanks{Y. M. Chee is with the Division of Mathematical Sciences,
		School of Physical and Mathematical Sciences,
		Nanyang Technological University,
		21 Nanyang Link,
		Singapore 637371 (e-mail: {\tt ymchee@ntu.edu.sg}).}
		\thanks{A. C. H. Ling is with the Department of Computer Science, 
		University of Vermont, Burlington, 
		Vermont 05405, USA (e-mail: {\tt aling@emba.uvm.edu}).}
		\thanks{J. Yin is with the 
		Department of Mathematics, Suzhou University, 215006 Suzhou, 
		People's Republic of China (e-mail: {\tt jxyin@suda.edu.cn}).}
}

\date{}                                           

\begin{document}

\maketitle

\begin{abstract}
\boldmath
Optimal partitioned cyclic difference packings (PCDPs) are shown to give
 rise to optimal frequency-hopping
sequences and optimal comma-free codes. New constructions for PCDPs, based on
almost difference sets and cyclic difference matrices, are given.
These produce new infinite families of optimal PCDPs (and hence optimal frequency-hopping
sequences and optimal comma-free codes). The
 existence problem for optimal PCDPs in ${\mathbb Z}_{3m}$, with
$m$ base blocks of size three, is also solved for all $m\not\equiv 8,16\pmod{24}$.
\end{abstract}


\section{Introduction}

\PARstart{F}{requency} hopping spread spectrum (FHSS) \cite{Pickholtzetal:1982}
is an important communication technique to
combat eavesdropping, Rayleigh fading, reduce interleaving depth and associated delay,
and enable efficient frequency reuse, giving rise to robust security and reliability.
As such, FHSS is
widely used in military radios, CDMA and GSM networks, radars and sonars, and
Bluetooth communications. 

In FHSS, an ordered list of frequencies, called a frequency-hopping sequence (FH sequence),
is allocated to each transmitter-receiver pair. Interference can occur when
two distinct transmitters use the same frequency simultaneously. In evaluating
the goodness of FH sequence design, the Hamming correlation function is
used as an important measure. Fuji-Hara {\em et al.} \cite{Fuji-Haraetal:2004}
introduced a new class of combinatorial designs and showed that they
are equivalent to FH sequences optimal with respect
to Hamming correlation. We call these combinatorial designs {\em
partitioned cyclic difference
packings} (PCDPs) in this paper.

PCDPs arise in another context.
In considering the construction of comma-free codes for synchronization over erroneous 
channels, Levenshte{\u\i}n \cite{Levenshtein:1971b} introduced {\em difference system
of sets} (DSS) and showed how DSS can be used to construct comma-free codes. We
establish connections between PCDP and DSS (and hence comma-free codes), especially
PCDPs that give rise to DSS and comma-free codes optimal with respect to redundancy.

As general results, we give new constructions of PCDPs via almost difference sets
and cyclic difference matrices. This gives new infinite families of optimal PCDPs.
The existence problem for optimal PCDPs in $\bbZ_{3m}$, with $m$ base blocks of size three,
is also solved for all $m\not\equiv 8,16\pmod{24}$.

\section{Mathematical Preliminaries}
\label{intro}

For a
positive integer $n$, the set $\{1,2,\ldots,n\}$ is denoted $[n]$, and
$\bbZ_n$ denotes the ring $\bbZ/n\bbZ$. The set $\bbZ_n\setminus\{0\}$ is
 denoted $\bbZ_n^\star$. The set of (nonzero) quadratic residues in $\bbZ_n$ is denoted
 $\bbZ_n^\square$ and the set of quadratic nonresidues of $\bbZ_n$ is denoted
 $\bbZ_n^\boxtimes$. For succinctness, we write $a_b$ for an element $(a,b)\in\bbZ_m\times\bbZ_n$.
 
 Given a collection
 \begin{equation*}
 \D = \{D_1, D_2, \ldots, D_m\}
 \end{equation*}
 of subsets (called {\em base blocks}) of $\bbZ_n$,
 define the {\em difference function} $\Phi_\D:\bbZ_n^\star\rightarrow\bbZ$ such that
 \begin{equation*}
  \Phi_\D(g) = \sum_{i=1}^m |(D_{i}+ g) \cap D_{i}|.
 \end{equation*}
 For positive integers $n$, $\lambda$, and multiset of positive integers $K$, a
 {\em cyclic difference packing} (CDP), or more precisely
 an $(n,K,\lambda)$-CDP, is
 a collection $\D$ of subsets of $\bbZ_n$ such that
 \begin{enumerate}[(i)]
 \item $K=[ |D| : D\in\D ]$; and
 \item $\lambda=\max_{g\in\bbZ_n^\star} \Phi_\D(g)$.
 \end{enumerate}
 If, in addition, $\D$ partitions $\bbZ_n$, then $\D$ is a
 {\em partitioned cyclic difference packing} (PCDP), or more precisely an
 $(n,K,\lambda)$-PCDP. For succinctness, we normally write the multiset $K$ in
 exponential notation: $[k_1^{a_1} k_2^{a_2} \cdots k_s^{a_s}]$ denotes the
 multiset containing $a_i$ occurrences of $k_i$, $i\in [s]$.
 The notion of PCDP is first introduced by Fuji-Hara {\em et al.} \cite{Fuji-Haraetal:2004}
 in their investigation of frequency-hopping sequences, where it is referred to as
 ``a partition type difference packing''.
 
It is not hard to verify that the following are equivalent definitions of a PCDP:
 \begin{enumerate}[(i)]
 \item $\D$ is an $(n,K,\lambda)$-PCDP if and only if $\D$ partitions $\bbZ_n$, and
 for any fixed $g\in\bbZ_n^\star$, the equation $x-y=g$ has at most
 $\lambda$ solutions $(x,y)\in \cup_{D\in\D} D\times D$.
 \item For a set $D\subseteq\bbZ_n$, let
 \begin{equation*}
 \Delta D =  \{a-b:a,b\in D, a\not=b\}.
 \end{equation*}
 Then $\D$ is an $(n,K,\lambda)$-PCDP if and only if $\D$ partitions $\bbZ_n$, and
 the multiset 
 \begin{equation*}
 \Delta \D = \bigcup_{i=1}^m \Delta D_i 
 \end{equation*}
 contains each element of $\bbZ_n$ at most $\lambda$ times.
 \end{enumerate}
 In the particular case where an $(n,K,\lambda)$-PCDP $\D$
 satisfies $\Phi_\D(g)=\lambda$ for all $g\in\bbZ_n^\star$, it is known as a {\em partitioned
 cyclic difference family} (PCDF), or more precisely an $(n,K,\lambda)$-PCDF.

 Given two positive integers $n$ and  $m < n$, it is obvious that
 any partition $\D$ of $\bbZ_n$ is an
 $(n, K, \lambda)$-PCDP for some $\lambda$. Furthermore, we have
 \begin{equation}
 \label{1.1}
 \lambda \geq    \left\lceil
 \frac{\sum_{i=1}^{m} |D_{i}|(|D_{i}| - 1)}{n-
                   1}\right\rceil    =  \left\lceil \frac{\sum_{i=1}^{m} |D_{i}|^2 -
                                       n}{n - 1}\right\rceil,
 \end{equation}
 since the multiset $\Delta\D$ contains
 $\sum_{i=1}^{m}|D_{i}|(|D_{i}| - 1)$ elements.
 The problem here we are concerned with is the
 construction of an $(n, K, \lambda)$-PCDP of $m$ base blocks with
 its {\em index} $\lambda$ as small as possible. Given $n$ and $m<n$, the minimum $\lambda$
 for which there exists an $(n,K,\lambda)$-PCDP of $m$ base blocks is
 denoted $\rho(n,m)$. An $(n,K,\lambda)$-PCDP of $m$ base blocks is {\em optimal} if
 $\lambda=\rho(n,m)$.
 
 From (\ref{1.1}), it is clear that
 \begin{equation} \label{1.2}
 \rho(n,m)  \geq  \left\lceil \frac{\sum_{i=1}^{m}|D_{i}|^2 -
                                       n}{n - 1}\right\rceil.
 \end{equation}
 The right side of (\ref{1.2}) cannot be
 determined uniquely by the parameters $n$ and $m$. To see when it
 attains the minimum for fixed $n$ and $m<n$, we write
 $n  = m\mu + \epsilon$ with $0 \le \epsilon \le m-1$. 
 It is well known that under the constraint $\sum_{i=1}^{m} |x_i|=n$, the sum
 $\sum_{i=1}^{m} x_i^2$ is minimized if and only if $|x_i-x_j|\leq 1$ for any
 $i,j\in[m]$. Hence, for an $(n,K,\lambda)$-PCDP $\D=\{D_1,\ldots,D_{m}\}$,
 the sum $\sum_{i=1}^{m} |D_i|^2$ attains the minimum if and only if $\D$
 contains exactly $\epsilon$ base blocks of size $\mu+1$ and $m-\epsilon$ base blocks
 of size $\mu$. Consequently, we have
  \begin{align} \label{1.3}
 \left\lceil \frac{\sum_{i=1}^{m} |D_{i}|^2 - n}{n - 1}\right\rceil 
   &= \left\lceil \frac{\epsilon (\mu + 1)^{2}+(m - \epsilon)\mu^{2} - m\mu - \epsilon}{n -
                   1}\right\rceil \nonumber \\
   & =  \left\lceil \frac{2\epsilon \mu + m\mu^2 - m\mu}{n -
                   1}\right\rceil \nonumber \\
                   & = 
                   \left\lceil \frac{2\epsilon \mu + (n-\epsilon)(\mu-1)}{n -
                   1}\right\rceil \nonumber \\
                  &=  \left\lceil \frac{(n-1)(\mu-1) + (\epsilon \mu + \epsilon + \mu-1)}{n -
                   1}\right\rceil \nonumber \\
                   & =  \left\lceil(\mu-1) +  \frac{\epsilon \mu + \epsilon + \mu-1}{n -
                   1}\right\rceil \nonumber \\
                   & =  \mu.
 \end{align}
The last equality in (\ref{1.3}) follows from the fact that $0 <
\epsilon \mu + \epsilon + \mu-1 \le n-1$, since $n = m\mu + \epsilon
> m$ and $0 \le \epsilon \le m-1$.

 It now follows from (\ref{1.2}) and (\ref{1.3}) that for any positive integers $m$ and $n= m\mu +
 \epsilon > m$ with $0 \le \epsilon \le m-1$,
 \begin{equation} \label{1.4}
 \rho(n,m)  \ge  \mu.
 \end{equation}

 We remark that for given positive integers $n$ and $m < n$, there may exist
 many optimal $(n, K, \lambda)$-PCDPs attaining the bound in
(\ref{1.4}). We also note that the lower bound on the
 function $\rho(n,m)$ in (\ref{1.4}) is not always attainable.

 The construction of optimal PCDPs have been studied by a number of authors.
 For more detailed information on PCDPs and known results, the reader is referred to
 \cite{Fuji-Haraetal:2004,Geetal:2006} and the references therein. In this paper, we
 make further investigation into optimal PCDPs. The paper is organized as follows. 
 In Section \ref{applications},
 we present the relationship among PCDPs,
 frequency-hopping sequences, and comma-free codes.
 Sections \ref{ADS}, \ref{CDM}, and \ref{PCDPs} are devoted to constructions of
 PCDPs, by which a number of new infinite classes of optimal PCDPs are
 produced. The existence of (optimal)  $(3m, [3^m],3)$-PCDPs is also determined
 for all $m\not\equiv 8,16\pmod{24}$.
 As a consequence, new infinite families of
 optimal frequency-hopping sequences and comma-free codes are
 obtained.

\section{Applications of PCDP}
\label{applications}

 PCDPs are closely related to frequency-hopping sequences and comma-free codes.
 We develop their relationship in this section.

 \subsection {PCDPs and Frequency-Hopping Sequences}

 Let $F = \{f_1, f_2, \ldots, f_{m}\}$ be a set of available
 frequencies, called a {\em frequency library}.
 As usual, $F^n$ denotes the set of
 all sequences of length $n$ over  $F$. An element of $F^n$ is called a 
 {\em frequency-hopping sequence} (FH sequence). Given two FH sequences
$X = (x_0, x_1, \ldots,$ $x_{n-1})$ and $Y = (y_0, y_1,
 \ldots,$ $y_{n-1})$, define their {\em Hamming correlation} $H_{X,Y}(t)$
 to be
 \begin{equation*}
 H_{X,Y}(t) = {\sum}_{i\in\bbZ_n}h[x_{i}, y_{i+t}], \quad t\in\bbZ_n,
 \end{equation*}
 where
 \begin{equation*}
 h[x, y]  = \begin{cases}
                   1, & \text{if $x = y$} \\
                   0, & \text{otherwise,}
                  \end{cases}
 \end{equation*}
 and all operations among position indices are performed in $\bbZ_n$.
 Further, define
 \begin{equation*}
 H(X) = {\rm max}_{t\in\bbZ_n^\star}H_{X,X}(t).
 \end{equation*}
 An FH sequence $X \in F^n$ is called {\em optimal} if $H(X) \leq
 H(X')$ for all $X' \in F^n$. Here we assume  that all
 transmitters use the same FH sequence, starting from different time
 slots. An FH sequence $X \in F^n$ with $H(X) = \lambda$ is
 called an $(n, m, \lambda)$-FH sequence.

 Frequency-hopping spread spectrum (FHSS)
  and direct sequence spread spectrum are two main spread coding technologies.
 In modern radar and
 communication systems, FHSS
 techniques have become very popular. FH sequences are used
 to specify which frequency will be used for transmission at any
 given time. Fuji-Hara {\em et al.} \cite{Fuji-Haraetal:2004} investigated
 frequency-hopping multiple access (FHMA) systems with a single
 optimal FH sequence using a combinatorial approach. They established the
 correspondence between frequency-hopping sequences and PCDPs. To be
 more precise, they labeled a frequency library $F$ of size $m$ by $\bbZ_m$
 and demonstrated that the set of position indices of
 an $(n, m, \lambda)$-FH sequence $X$  gives an $(n, K,
 \lambda)$-PCDP where $\lambda = H(X)$, and vise versa.
 We state this correspondence  in the following theorem using our notations.

\vskip 10pt
 \begin{theorem}[Fuji-Hara {\em et al.} \cite{Fuji-Haraetal:2004}]
 \label{char-1}
 There exists an optimal  $(n, m, {\lambda})$-FH sequence $X$ over the set of frequencies
 $F = \bbZ_m$ if and only if
 there exists an optimal $(n, K, {\lambda})$-PCDP of $m$ base blocks.
 \end{theorem}
 \vskip 10pt

 Theorem \ref{char-1} reveals that in order to construct optimal FH
 sequences, one
 needs only to construct optimal PCDPs. This  serves
 as the motivation behind our consideration of PCDPs.

 \subsection {PCDPs and Comma-Free Codes}

 Consider the process of transmitting data over a channel, where the data being
 sent is a stream of symbols from an alphabet $Q$ of
 size $q$.  The data stream consists
 of consecutive messages, each being a sequence of $n$ consecutive
 symbols:
 \begin{equation*}
 \cdots \underbrace{x_1 \cdots x_n}\underbrace{y_1 \cdots y_n}\cdots
 \end{equation*}
 The synchronization problem that arises at the receiving end is the
 task of correctly partitioning the data stream into messages of length
 $n$, as opposed to incorrectly conceiving a sequence of $n$ symbols
 that is the concatenation of the end of one message with the beginning
 of another message as a single message:
 \begin{equation*}
 \cdots \underbrace{x_{i+1} \cdots x_n}\underbrace{y_1 \cdots y_i}\cdots 
 \end{equation*}

 One way to resolve the synchronization problem uses comma-free codes.
 A code is a set $\C \subseteq Q^n$, with its elements called {\em codewords}.
 $\C$ is termed a {\em comma-free code} if the
{\em  concatenation}
\begin{equation*}
 T_i(x, y)= x_{i+1} \cdots x_ny_1 \cdots y_i
 \end{equation*}
 of any two not necessarily distinct  codewords $x=(x_1 \cdots x_n)$ and $y = (y_1
 \cdots y_n)$ is never a codeword. More
 generally, associated with a code $\C \subseteq Q^n$, one can define
 its {\em comma-free index} $I(\C)$ as
 \begin{equation*}
 I(\C) = \min_{\text{$x,y,z\in\C$ and $i\in[n-1]$}} d_H(T_i(x, y), z),
 \end{equation*}
 where $d_H(\cdot,\cdot)$ denotes the Hamming distance function. If 
 $I(\C)>0$, then $\C$ is a comma-free code,
 and hence we can distinguish a codeword from a concatenation of two
 codewords even in the case when up to $\lfloor(I(\C)-1)/2\rfloor$ errors
 have occurred \cite{Golombetal:1958,Tonchev:2005}.

 Codes with prescribed comma-free index can be constructed by using
 {\em difference systems of sets} (DSS), a combinatorial structure
 introduced by Levenshte{\u\i}n \cite{Levenshtein:1971b}
 (see also \cite{Levenshtein:2004,Tonchev:2005}).
 An $(n, K, \eta)$-DSS is a collection $\F= \{D_1,D_2, \ldots, D_{m}\}$ of
 $m$ disjoint subsets of $\bbZ_n$ such that the multiset
 \begin{equation*}
 \{a - b \pmod{n} : \text{$a \in D_i$, $b \in D_j$, $i,j\in [m]$ and $i\not=j$}\}
 \end{equation*}
 contains each element of $\bbZ_n^\star$ at least $\eta$ times,
 where $K = [|D| : D\in \F]$. Application of DSS
 to code synchronization
 requires that the {\em redundancy}
 \begin{equation*}
 r_m(n, \eta) = \sum_{i=1}^{m}|D_i|
 \end{equation*}
 be as small as possible. Levenshte{\u\i}n \cite{Levenshtein:1971b} proved that
 \begin{equation} \label{Leve}
 r_m(n, \eta) \ge  \sqrt{\frac{\eta m(n - 1)}{m - 1}}.
 \end{equation}
For more detailed information on comma-free codes, the reader is
referred to \cite{Levenshtein:2004,Tonchev:2005} and the
 references therein. Here, we are interested in  the link between
 PCDP and DSS, which is stated in the following theorem.

\vskip 10pt
 \begin {theorem}
 \label {Link} 
 Let $m$ and $n = m\mu+ \epsilon > m$ be positive integers
 and  $0 \le \epsilon \le m-1$. If an $(m\mu + \epsilon, [(\mu+1)^{\epsilon}\mu^{m - \epsilon}],
 \mu)$-PCDP exists, then so does an $(n, [(\mu+1)^{\epsilon}\mu^{m - \epsilon}], \eta)$-DSS
 of minimum redundancy $n$, where $\eta = n - \mu = (m-1)\mu + \epsilon.$
 \end {theorem}

 \begin{proof}
 Let $\D= \{D_{1}, D_{2}, \ldots, D_{m}\}$ be
 an $(n, [(\mu+1)^{\epsilon}\mu^{m - \epsilon}],
 \mu)$-PCDP. By definition of a PCDP, we have
 \begin{equation*}
 \mu = \max_{g\in\bbZ_n^\star} \Phi_\D(g).
 \end{equation*}
Let us define $\Gamma_\D:\bbZ_n^\star\rightarrow\bbZ$ such that
 \begin{equation*}
 \Gamma_\D(g) =\sum_{i,j\in[m], i\not=j}|(D_i + g)
 \cap D_j|.
 \end{equation*}
 It follows that
 \begin{equation*}
 \Phi_\D(g) + \Gamma_\D(g) = |(\bbZ_n + g)\cap \bbZ_n| = n,
 \end{equation*}
 for any $g \in\bbZ_n^\star$. Furthermore,
 \begin{align*}
 \min_{g\in\bbZ_n^\star} \Gamma_\D(g)
 &=  \min_{g\in\bbZ_n^\star} (n - \Phi_\D(g)) \\
 &=  n - \max_{g\in\bbZ_n^\star} \Phi_\D(g) \\
 & =  n - \mu.
 \end{align*}
 
 Hence, $\D$ is an $(n, [(\mu+1)^{\epsilon}\mu^{m - \epsilon}],\eta)$-DSS for
 $\eta=n-\mu$ and with redundancy $r_m(n,\eta)=n$.
 We now prove this redundancy to be minimum. Since $\eta = n - \mu$ and $n = m\mu+ \epsilon$,
 the right side of the inequality (\ref{Leve}) equals
\begin{align*}
    \sqrt{\frac{\eta m(n - 1)}{m - 1}}
  &=\sqrt{\frac{m(n-\mu)(n-1)}{m-1}}\\
  &=\sqrt{\frac{m(n-1)n - m\mu(n-1)}{m-1}}\\
  &=\sqrt{\frac{m(n-1)n - (n-\epsilon)(n-1)}{m-1}}\\
 &=\sqrt{\frac{(m-1)(n-1)n + \epsilon(n-1)}{m-1}}.
 \end{align*}
 This implies
 \begin{equation*}
 n - 1 < \sqrt{\frac{\eta m(n - 1)}{m - 1}} < n,
 \end{equation*}
 since $0 \le \epsilon \le (n-1)(m-1).$  
 \end{proof}
 \vskip 10pt

 Theorem \ref {Link} shows that the DSS derived from a PCDP of
 minimum index has minimum redundancy, and hence produces optimal
 comma-free codes with respect to the bound (\ref{Leve}). This  serves
 to provide another motivation behind the study of PCDPs.

 \section{Constructions from Almost Difference Sets}
 \label{ADS}

  An {\em almost difference set} in an additive group $G$  of order $n$,
  or an $(n, k, \lambda; t)$-ADS in short, is 
  a $k$-subset $D$ of $G$ such that
  the multiset
 $\{a - b: a, b \in D\ \mbox{and}\ a\not= b\}$ contains $t$ nonzero elements of $G$, each
 exactly $\lambda$ times, and each of the remaining $n-1-t$ nonzero elements
 exactly $\lambda + 1$ times.
 This is equivalent to saying that
 \begin{equation*}
 \Phi_{\{D\}}(g) \:= |(D + g) \cap D|.
 \end{equation*}
 takes on the value $\lambda$ exactly $t$ times and the value
 $\lambda + 1$ exactly $n - 1 - t$ times,  when $g$ ranges over
 all the nonzero elements of $G$. An obvious necessary condition
 for the existence of a $(n, k, \lambda; t)$-ADS is
 \begin{equation*}
 (\lambda + 1)(n-1) \equiv t \pmod {k(k-1)}.
 \end{equation*}
 In the extreme case where $t = n - 1$,
 an $(n, k, \lambda; t)$-ADS is  an $(n, k, \lambda)$ {\em difference set}
 in the usual sense (see \cite{Jungnickeletal:2007}).
 It should be apparent to the reader that 
 an $(n, k, \lambda; t)$-ADS in $\bbZ_n$ is an
 $(n, \{k\}, \lambda+1)$-CDP.

In this section, we construct new optimal PCDPs from almost difference sets.
We begin with the following result.

\vskip 10pt
 \begin {proposition}
 \label {Lem2k}
 Let $n = 2\mu$ be a positive integer and let $D$ be a $\mu$-subset of $\bbZ_n$.
 Let $\widehat{D} = \bbZ_n \setminus D$.
 If one of $D$ and $\widehat{D}$ is an $(n, \mu, \lambda; t)$-ADS in $\bbZ_n$,
 then so is the other.
 \end {proposition}
 
 \begin{proof}
We need only prove that
\begin{equation*}
 \Phi_{\{D\}}(g) = \Phi_{\{\widehat{D}\}}(g)
 \end{equation*}
 for any $g \in \bbZ_{2\mu}^{\star}.$ In fact, since  $\{D, \widehat{D}\}$ is a partition of $\bbZ_{2\mu}$,
 and $|D| = |\widehat{D}|= \mu$, we have
 \begin{equation*}
 |D \cap (D + g)| + |D \cap (\widehat{D} + g)| = \mu
 = |D \cap (\widehat{D} + g)| + |\widehat{D} \cap (\widehat{D} + g)|. 
 \end{equation*}
 Hence,
 \begin{equation*}
 \Phi_{\{D\}}(g)= \mu - |D \cap (\widehat{D} + g)| =
 \Phi_{\{\widehat{D}\}}(g).
 \end{equation*}
 This equality does not depend
 on the choice of $g\in\bbZ_{2\mu}^\star$. 
 \end{proof}
 \vskip 10pt

 As an immediate consequence of Proposition \ref{Lem2k}, we have the following corollary.

\vskip 10pt
 \begin {corollary}
 \label {Lem2k-cor} 
 Let $n=2\mu$, $n\equiv 2\pmod{4}$.
 Then  $\rho(n,2)$ cannot attain the lower bound $\mu$ in (\ref{1.4}), that is,
 $\rho(n,2) \ge \mu + 1 = (n+2)/2$.
 \end {corollary}

 \begin{proof}
  By assumption, $n= m\mu + \epsilon$ with $m= 2, \mu =
 n/2$ and $\epsilon = 0$. Since $n \equiv 2\pmod{4}$,
 $\mu$ is odd.  On the other hand, for any $(2\mu, [\mu^2], \lambda)$-PCDP $\D=
 \{D_1, D_2\}$, we have
 \begin{equation*}
 \Phi_\D(g) = \Phi_{\{D_1\}}(g) + \Phi_{\{D_2\}}(g) = 2\Phi_{\{D_1\}}(g)
 \end{equation*}
 by Proposition \ref{Lem2k}. So, $\lambda$
  must be even.  Hence,  an $(n, [(n/2)^2], \mu)$-PCDP cannot
  exist, which implies $\rho(n,2) \ge \mu + 1 =
 (n+2)/2$.  
 \end{proof}
 \vskip 10pt

 Now we turn to constructions.

\vskip 10pt
 \begin {proposition}
 \label {AdsCon-1} 
 Let $n \equiv 0\pmod{4}$.
 If there exists an $(n, n/2, (n-4)/4; n/4)$-ADS in $\bbZ_n$, then
 there exists an optimal $(n, [(n/2)^2], n/2)$-PCDP.
 \end {proposition}
 
 \vskip 10pt
 \begin{proof}
    Let $D$ be the given
 $(n, n/2, \lambda; t)$-ADS in $\bbZ_n$, where $\lambda = (n-4)/4$ and $t = n/4$.
Let $\widehat{D}=\bbZ_n\setminus D$.
Then $\D = \{ D,\widehat{D}\}$ is a partition of $\bbZ_n$ and as
in the proof of Proposition \ref{Lem2k}, we
 have
\begin{align*}
 \Phi_\D(g) & =  |D \cap (D + g)| + |\widehat {D} \cap (\widehat {D} + g)|\\
                 & = \Phi_{\{D\}}(g) + \Phi_{\{\widehat {D}\}}(g)\\
                 & =  2\Phi_{\{D\}}(g).
 \end{align*}
 Hence, for any $g \in \bbZ_n^\star$, 
 $\Phi_\D(g)$ takes on the value $(n-4)/2$ exactly $t = n/4$ times
 and takes on the value $n/2$ exactly $n- 1 - t = (3n-4)/4$ times.
 Therefore, $\D$ is an $(n, [(n/2)^2], n/2)$-PCDP. Its optimality follows immediately
 from (\ref{1.4}).   
 \end{proof}
 
\vskip 10pt
 \begin {proposition}
 \label {AdsCon-2} 
 Let $n \equiv 2\pmod{4}$.
 If there exists an $(n, n/2, (n-2)/4; (3n-2)/4)$-ADS in $\bbZ_n$, then there exists
 an optimal $(n, [(n/2)^2], (n+2)/2)$-PCDP.
 \end {proposition}
 
 \begin{proof}
 Employing the same technique as in the proof of Proposition \ref{AdsCon-1}, we
 form an $(n, [(n/2)^2], (n+2)/2)$-PCDP. The fact that its index attains the minimum
 follows from Corollary \ref{Lem2k-cor}.    
 \end{proof}
\vskip 10pt

 Almost difference sets in Abelian groups have been well studied in terms of
 sequences with optimal autocorrelation \cite{Lempeletal:1977,Arasuetal:2001} and are
 known to exist for certain parameters $n, k, \lambda$ and $t$.
 Before stating the known
 results on almost difference sets in $\bbZ_n$, some
 terminologies from finite fields are needed. Let $q$ be a prime
 power. The finite field of $q$ elements is denoted $\bbF_q$. Let ${\omega}$ be a
 primitive element of $\bbF_q$. For $e$ dividing $q-1$, define
 $D_i^{(e,q)} = {\omega}^{i}\langle\omega^{e}\rangle$,
 where $\langle\omega^{e}\rangle$ is the unique multiplicative subgroup of $\bbF_q$
 spanned by $\omega^{e}$. For $0 \le  h \neq r \le e-1$, define
 \begin{equation*}
    (h, r)_e = |(D_h^{(e,q)} + 1) \cap D_r^{(e,q)}|.
 \end{equation*}
 These constants $(h, r)_e$ are known as {\em cyclotomic numbers
 of order $e$}. The number $(h, r)_e$ is the number of solutions to the equation
$x + 1 = y$, where $x \in D_h^{(e,q)}$ and $y \in D_r^{(e,q)}$. The following results are known.

\vskip 10pt
 \begin {proposition}[Lempel {\em et al.} \cite{Lempeletal:1977}]
 \label{Ads-1}
 Let $q$ be an odd prime power and let $D = \log_{\omega}(D_1^{(2,q)} - 1)$. Then
 \begin {enumerate}[(i)]
 \item $D$ is a $(q - 1, (q-1)/2, (q-3)/4,(3q-5)/4)$-ADS in $\bbZ_{q-1}$, provided
             $q \equiv 3\pmod{4}$;
 \item $D$ is a $(q - 1, (q-1)/2, (q-5)/4,(q-1)/4)$-ADS in $\bbZ_{q-1}$ provided
             $q \equiv 1\pmod{4}$.
 \end {enumerate}
 \end {proposition}

\vskip 10pt
 \begin {proposition}[Ding {\em et al.} \cite{Dingetal:2001}]
 \label{Ads-2}
 Let $p \equiv  5\pmod{8}$ be an odd prime. It is known that
 $p = s^2 + 4t^2$ for some $s$ and $t$ with $s \equiv \pm 1\pmod{4}$.
 Set $n = 2p$. Let $i, j, l \in \{0, 1, 2, 3\}$ be three pairwise
 distinct integers, and
 \begin{align*}
 D_{(i,j,l)} = & \left[\{0\} \oplus (D_i^{(4, p)}\cup D_j^{(4, p)})\right]  \bigcup \\
 & \left[\{1\} \oplus (D_l^{(4, p)}\cup D_j^{(4,p)})\right]  \bigcup \{(0, 0)\}.
 \end{align*}
 Then $D_{(i, j, l)}$ is an $(n, n/2, (n-2)/4,(3n-2)/4)$-ADS in $\bbZ_{2} \oplus \bbZ_p$, being
             isomorphic to $\bbZ_{2p}$ when
 \begin{enumerate}[(i)]
 \item $t = 1$ and  $(i,j, l) \in \{(0, 1, 3)$, $(0, 2, 3)$, $(1, 2, 0)$, $(1, 3, 0)\}$
 or
 \item $s = 1$ and  $(i, j, l) \in \{(0, 1, 2)$, $(0, 3, 2)$, $(1, 0, 3)$, $(1, 2, 3)\}$.
 \end{enumerate}
 \end{proposition}
 \vskip 10pt

 Combining the results of Propositions \ref{AdsCon-1}--\ref{Ads-2} gives
 us new optimal PCDPs  as follows.

\vskip 10pt
 \begin {theorem}
 \label{PDHPs-1}
There exist
 \begin {enumerate}[(i)]
 \item an optimal $(q-1, [(q-1/2)^2], (q-1)/2)$-PCDP for any prime
 power  $q \equiv 1\pmod{4}$; and
 \item an optimal $(n, [(n/2)^2], (n+2)/2)$-PCDP if $n = 2p$ or
 $n = q - 1$ where $q \equiv 3\pmod{4}$ is a prime power and
 $p \equiv 5\pmod{8}$ is a prime.
 \end {enumerate}
 \end{theorem}

 \section{Constructions from Cyclic Difference Matrices}
 \label{CDM}

     Consider a $k \times n$ matrix $\vM = (m_{ij})$, $i\in[k]$, $j\in[n]$,
whose entries are taken from an additive group $G$ of order $n$.   If for any two
 distinct row indices $r,s\in[k]$, the
 differences
$m_{rj} - m_{sj}$, $j\in[n]$, comprise all the elements of
$G$,
 then the matrix $\vM$ is said to be an $(n, k, 1)$ {\em difference matrix} (DM), or
 an $(n, k, 1)$-DM over $G$. 
 
 Our constructions require a difference
 matrix over the cyclic group of order $n$,
 that is $G = \bbZ_n$. In this case, the difference matrix is called {\em cyclic} and is denoted
 by $(n,k,1)$-CDM. A cyclic DM is {\em normalized} if all entries in its first row and
 first column are zero. The property of a cyclic DM
 is preserved even if we add an element of $\bbZ_n$ to all entries in
 any row or column of the matrix. Hence, without loss of generality,
 one can always assume that a cyclic DM is normalized. If we delete the
 first row from a normalized $(n,k,1)$-CDM, then we obtain a {\em derived} $(n,k-1,1)$-CDM,
 each of whose rows forms a permutation on $\bbZ_n$. 
 We adopt the
 terminology used in
 \cite{Fuji-Haraetal:2004} and call the derived $(n,k-1,1)$-CDM {\em homogeneous}. In a homogeneous
 cyclic DM, every row forms a permutation on the elements of $\bbZ_n$ and the
 entries in the first column are all zero. From the point of view of
 existence, a $(n,k,1)$-CDM is obviously equivalent to a homogeneous
 $(n,k-1,1)$-CDM, and we use the terms $(n,k,1)$-CDM and
homogeneous $(n,k-1,1)$-CDM interchangeably.

 Difference matrices have attracted considerable attention in
 design theory, since they can often be used as building blocks for
 other combinatorial objects.
The multiplication table of
 $\bbF_q$ constitutes a normalized $(q, q, 1)$-DM.
 When $q$ is a prime, it is a normalized $(q, q, 1)$-CDM. Hence, a  homogeneous
 $(q, q-1, 1)$-CDM exists for any prime $q$. Deleting $q - 1-k$ rows
 from this cyclic DM
produces a homogeneous $(q, k,1)$-CDM. We record this fact below.

\vskip 10pt
 \begin {proposition}
 \label {CDM-01} 
 Let $p$ be a prime and $k$ an integer satisfying $2\le k \le p-1$. Then there exists
 a  homogeneous $(p, k, 1)$-CDM.
 \end {proposition}
 \vskip 10pt

 The following product construction for cyclic DMs  is known
 (see, for example, \cite{Ge:2005,Yin:2005}).

\vskip 10pt
 \begin {proposition}
 \label {CDM-02} 
 If a  homogeneous $(n_1, k, 1$)-CDM and a  homogeneous $(n_2, k, 1$)-CDM both exist,
 then so does a  homogeneous $(n_1n_2, k, 1$)-CDM.
 \end {proposition}
 \vskip 10pt

Proposition \ref {CDM-01} and Proposition \ref {CDM-02} now gives the
 following existence result.

\vskip 10pt
\begin {proposition}
\label {CDM-1}
 Let $n \geq 3$ be an integer whose prime factors are at least the prime $p$.
 Then for any integer $k$ satisfying  $2\le k \le p-1$, a  homogeneous
 $(n, k, 1$)-CDM exists.
 \end {proposition}
 \vskip 10pt

 We also need the following result.

\vskip 10pt
 \begin{proposition}[Ge \cite{Ge:2005}]
 \label {CDM-2}
 Let $n \geq 5$ be an odd integer with $\gcd(n, 27)\not=9$. Then there exists a  homogeneous
 $(n, 3, 1)$-CDM.
 \end {proposition}
\vskip 10pt

Now we develop our constructions to obtain optimal PCDPs from cyclic
DMs.

\vskip 10pt
 \begin {theorem}
 \label{Con-1} 
 Let $m$ and $\mu$ be two positive integers. If a homogeneous
 $(m, \mu, 1$)-CDM exists, then  so does an  optimal $(m\mu, [\mu^m], \mu)$-PCDP.
 \end {theorem}
 
 \begin{proof}
   Let  $\vM = (m_{ij})$, $i\in[\mu]$, $j\in[m]$, be
 a homogeneous
 $(m, \mu, 1$)-CDM over $\bbZ_{m}.$
 From $\vM$ we construct another $\mu \times m$ matrix $\vR$ whose
 entries are taken from
$\bbZ_{m\mu}$ by replacing every entry $m_{ij}$ of $\vM$ with
 $i-1 + m_{ij}\mu$, $i\in[\mu]$, $j\in[m]$. Write $D_j$ for the $\mu$-subset of $\bbZ_{m\mu}$
consisting of the elements on the $j$-th column of $\vR$, $j\in[m]$. Write
\begin{equation*}
\D= \{D_{1}, D_{2}, \ldots, D_{m}\}.
\end{equation*}
 Then the properties of a homogeneous
 $(m, \mu, 1$)-CDM guarantee the following conclusions:
 \begin{itemize}
 \item $\D$ partitions $\bbZ_{m\mu}$.
 \item Let $H = \mu \bbZ_m \:= \{j\mu: 0\leq j\leq m-1\}$
 be the unique additive subgroup of order $m$ in $\bbZ_{m\mu}$.
 Then, for any nonzero element $g \in \bbZ_{m\mu}$, we have
 \begin{equation*}
 \Phi_\D(g) = \begin{cases}
        0, & \text{if $g\in H$} \\
        \mu, & \text{otherwise}.
\end{cases}
\end{equation*}
\end{itemize}
Therefore, $\D$ is an $(m\mu, [\mu^m], \mu)$-PCDP, and it is
optimal, since its index meets the bound in (\ref{1.4}).  
\end{proof}
\vskip 10pt

Applying Theorem \ref{Con-1} and Proposition \ref{CDM-1}, we obtain the
following new infinite family of optimal PCDPs.

\vskip 10pt
\begin {theorem}
\label {ConA-1} 
 Let $m \geq 3$ be an integer whose prime factors are not less than prime
 $p$.
 Then for any integer $\mu$ satisfying  $2\le \mu \le p-1$,
 an optimal $(m\mu, [\mu^m], \mu)$-PCDP exists.
\end {theorem}

\vskip 10pt
\begin{example}
\label {EX-1} 
 In Theorem \ref{Con-1}, take $m = 7$, $\mu = 3$, and consider the homogeneous
 $(m, \mu, 1$)-CDM
 \begin{equation*}
 \vM =
    \begin{bmatrix}
        0&1&2&3&4&5&6\\
        0&2&4&6&1&3&5\\
        0&3&6&2&5&1&4\\
\end{bmatrix}.
\end{equation*}
Replace each entry $m_{ij}$ of $\vM$ with
 $i-1 + 3m_{ij}$, $i\in[3]$, $j\in[7]$, to obtain the
$3 \times 7$ matrix over $\bbZ_{21}$
 \begin{equation*}
 \vR =
    \begin{bmatrix}
        0&3 &6  &9 &12&15&18\\
        1&7 &13 &19&4 &10&16\\
        2&11&20 &8 &17&5 &14\\
\end{bmatrix}.
\end{equation*}
 Finally, take the columns of $\vR$ as base blocks over $\bbZ_{21}$:
 \begin{equation*}
  \begin{array}{llllllllllllllllll}
        D_1=\{0, 1, 2\} &  D_2=\{3, 7, 11\} &  D_3=\{6, 13, 20\} \\
        D_4=\{8, 9, 19\} & D_5=\{4, 12, 17\} & D_6=\{5, 10, 15\} \\
         D_7=\{14, 16, 18\} & \\
\end{array}
\end{equation*}
It is readily checked that  $\D = \{D_{1}, D_{2}, \ldots,
D_{7}\}$ is an optimal $(21, [3^7], 3)$-PCDP, as desired. 
 \end{example}
 \vskip 10pt

The following result is a variant of Theorem \ref{Con-1}.

\vskip 10pt
 \begin {proposition}
 \label{Con-2} 
 Let $p$ be an odd prime. Then an  optimal $(p^2, [\mu^{m-1}(\mu+1)^1], \mu)$-PCDP
 exists, where $m = p+1$ and $\mu = p-1$.
 \end {proposition}
 
 \begin{proof}
   By Proposition \ref{CDM-01}, there exists 
   $\vM = (m_{ij})$, $i\in[p-1]$, $j\in[p]$, which is a homogeneous
 $(p, p-1, 1$)-CDM over $\bbZ_{p}$.
 As in the proof of Theorem \ref{Con-1}, we construct a $(p-1) \times p$ matrix $\vR$ whose
 entries are taken from
$\bbZ_{p^2}$ by replacing every entry $m_{ij}$ of $\vM$ with
 $i + m_{ij}p$, $i\in[p-1]$, $j\in[p]$. Then we write $D_j$ for the $(p-1)$-subset of $\bbZ_{p^2}$
consisting of the elements on the $j$-th column of $\vR$ for $j\in[p]$. Then
\begin{equation*}
\D= \{D_{1}, D_{2}, \ldots, D_{p}\}
\end{equation*}
 is a cyclic difference packing. Observe that the rows of $\vR$ are indexed by the elements of
 $\bbZ_{p}^\star$. Hence, $\D$ is a partition of $\bbZ_{p^2}\setminus p\bbZ_p$. 
 On the other hand, we have
 \begin{equation*}
 \Phi_\D(g) =\begin{cases}
        0, & \text{if $g \in p\bbZ_p$} \\
        p-2, & \text{otherwise}.
\end{cases}
\end{equation*}
For the desired PCDP, let $\widehat{D_1} = D_1 \cup \{0\}$ and
$D_{p+1}= \{jp: j\in[p-1]\}.$  It turns out that
\begin{equation*}
\F = (\D\setminus \{D_1\})\bigcup \{\widehat{D_1}, D_{p+1}\}
\end{equation*}
is a $(p^2, [\mu^{m-1}(\mu+1)^1], \mu)$-PCDP. Its
optimality is straightforward to verify.   
\end{proof}

\vskip 10pt
\begin {example}
\label {EX-2}
 In Proposition \ref{Con-2}, take $p = 5$ (and hence $m=6$ and $\mu = 4$),
 and consider the homogeneous
 $(5, 4, 1$)-CDM 
 \begin{equation*}
 \vM =
    \begin{bmatrix}
        0&1&2&3&4\\
        0&2&4&1&3\\
        0&3&1&4&2\\
        0&4&3&2&1\\
\end{bmatrix}.
\end{equation*}
Its corresponding $4 \times 5$ matrix $\vR$ over $\bbZ_{25}$ is given by
 \begin{equation*}
 \vR =
    \begin{bmatrix}
        1&6 &11&16&21\\
        2&12&22&7 &17\\
        3&18&8 &23&13\\
        4&24&19&14&9 \\
\end{bmatrix}.
\end{equation*}
Then an optimal $(25, [4^55^1], 4)$-PCDP is formed by the
following base blocks over $\bbZ_{25}$:
\begin{equation*}
  \begin{array}{lllllllllllll}
        \widehat{D_0}=\{0, 1, 2, 3, 4\} & D_1=\{6, 12, 18, 24\} \\
        D_2=\{8, 11, 19, 22\} & D_3=\{7, 14, 16, 23\} \\
         D_4=\{9, 13, 17, 21\} & D_5=\{5, 10, 15, 20\}
\end{array}
\end{equation*}
 \end{example}
 \vskip 10pt

 Based on Proposition \ref{Con-2}, we can establish the following new infinite series of optimal
 PCDPs.

\vskip 10pt
 \begin {theorem}
 \label{Con-3} 
 Let $p$ be an odd prime and $n \geq 2$.
 Then   an  optimal $(p^n, [\mu^{m-1}(\mu+1)^1], \mu)$-PCDP exists, where
 $m = \frac{p^n-1}{p-1}$ and $\mu = p -1$.
 \end {theorem}
 
 \begin{proof}
The proof is by induction on $n$.
  If $n = 2$, the conclusion holds by Proposition \ref{Con-2}.  Now suppose that the assertion is true
  when $n = k \ge 3$. Consider the case $n = k+1$. From Proposition \ref {ConA-1}, we know that  a homogeneous
 $(p^k, p-1, 1$)-CDM exists. Employing the same technique
  as in the proof of Proposition \ref{Con-2}, from this CDM we can
 form a collection
\begin{equation*}
 \D= \{D_{1}, D_{2}, \cdots, D_{p^k}\}
\end{equation*}
 of $(p-1)$-subsets of $\bbZ_{p^{k+1}}$ in such a way that
 \begin{itemize}
 \item $\D$ partitions $\bbZ_{p^{k+1}}\setminus p\bbZ_{p^k}$, and
 \item for any $g \in \bbZ_{p^{k+1}}^\star$,
 \begin{equation*}
 \Phi_\D(g) =\begin{cases}
        0, & \text{if $g \in p\bbZ_{p^k}$} \\
        p-2, & \text{otherwise}.
\end{cases}.
\end{equation*}
\end{itemize}
 Since $p\bbZ_{p^k}$ is isomorphic to $\bbZ_{p^k}$, 
 by our induction hypothesis we can construct an optimal PCDP
 ${\F}$ of index $p-1$ in
 $p\bbZ_{p^k}$, which has exactly $\frac{p^{k}-1}{p-1} - 1$
 base blocks of size $p-1$ and one base block of size $p$.
 It can be checked that $\D\cup \F$ is an optimal
 $(p^{k+1}, [(p-1)^{\frac {p^{k+1}-1}{p-1}-1}p^1], p-1)$-PCDP.   
 \end{proof}

\section{The Existence of $(3m, [3^m], 3)$-PCDPs}
\label{PCDPs}

In this section, the existence of $(3m,[3^m],3)$-PCDP is settled
for all $m\not\equiv 8,16\pmod{24}$. PCDPs with such
parameters are optimal.
Our proof technique requires a generalization of cyclic difference matrices.

Let $G$ be a cyclic group of order $n$ containing a subgroup $H$
of order $w$. A $k \times (n-w)$ matrix $\vM=(m_{ij})$, $i\in[k]$, $j\in[n-w]$,
with entries from $\bbZ_{n}$ is said to be a {\em holey} DM if
for any two distinct row indices $r$ and $s$ of $\vM$, $r,s\in[k]$, the differences
$m_{rj}-m_{sj}$, $j\in[n-w]$, comprise all the element of $G \setminus H$. 
For convenience, we refer to such a matrix
$\vM$ as an $(n,k,1;w)$-HDM over $(G;H)$, or simply an $(n,k,1;w)$-HDM when $G$ and
$H$ are clear from the context. $H$ is the {\em hole} of the holey DM.

The property of a holey DM is preserved even if we
add any element of $\bbZ_{n}$ to any column of the matrix.  Hence, without loss of
generality, one can always assume that the all entries in the first row
of an holey DM are zero. If we delete the first row from such a
holey DM, then we obtain an $(n,k-1,1;w)$-HDM, where the entries of a row consist
of all the elements of $G \setminus H$, and we term the derived
$(n,k-1,1;w)$-HDM {\em homogeneous}.  Because of this equivalence, we use the
terms $(n,k,1;w)$-HDM and homogeneous $(n,k-1,1;w)$-HDM interchangeably.

We introduce one more object which is crucial to the
construction for PCDPs in this section. Let $G$ be a cyclic group of
order $n$. A {\em partial} DM of order $n$ (denoted PDM$(n)$) 
 is a $3 \times (n-3)$ matrix $\vM=(m_{ij})$ with entries from
$G$ such that the entries in each row of $\vM$ are distinct,
and for any two distinct row indices $r,s\in[3]$, 
the differences $m_{rj}-m_{sj}$, $j\in[q-3]$, contains each element of $G$ at most once.
In addition, if the three sets of missing elements $D_i=G
\setminus \cup_{j=1}^{q-3}\{m_{ij}\}$, $i\in[3]$, and the multiset of
differences $\cup_{i=1}^{3}\{x-y: \text{$x,y \in D_i$ and $x \neq y$}\}$
contain each element of $G$ at most three times, we
called the partial DM {\em extendible}.

\vskip 10pt
\begin{example} 
\label{pdm8}
An extendible PDM$(8)$ over $\bbZ_{8}$:
\begin{equation*}
\begin{bmatrix}
        2&3&6&7&4\\
        1&6&3&5&4\\
        7&1&5&2&0 \\
\end{bmatrix}.
\end{equation*}
The three sets of missing elements are 
$\{0,1,5\}$, $\{0,2,7\}$, and $\{4,3,6\}$. 
\end{example}
\vskip 10pt

The following proposition gives the connection between extendible
 partial DMs and $(3m,[3^m],3)$-PCDPs.

\vskip 10pt
\begin{proposition} 
\label{pdmpcdp} 
Suppose there exists an extendible PDM$(m)$. Then there exists a
$(3m,[3^m],3)$-PCDP.
\end{proposition}

\begin{proof}
For each column $[a,b,c]^T$ of the PDM$(m)$, we
construct a base block $\{3c,3b+1,3a+2\}$ of the PCDP. If
$D_i=\{a,b,c\}$ is the 
set of missing element in row $i$, $i\in[3]$, we
construct a base block $\{3a+3-i,3b+3-i,3c+3-i\}$ of the PCDP. This gives 
a total of $m$ base blocks. It is easy to check that the
conditions of an extendible partial DM ensure that the base
blocks form a PCDP. 
\end{proof}
\vskip 10pt

The usefulness of holey DMs stems from the fact that they can be used to produce large
extendible partial DMs by ``filling in'' the hole of a holey DM with a smaller extendible partial
DM.

\vskip 10pt
\begin{proposition}[Filling in Hole]
\label{compose}
Suppose there exist a homogeneous
$(m,3,1;w)$-HDM and an extendible PDM$(w)$.
Then there exists an extendible PDM$(m)$.
\end{proposition}

\begin{proof}
  Multiple each entry of the extendible PDM$(w)$ by 
$m/w$ and add the columns of 
the resulting matrix to the  
homogeneous $(m,3,1;w)$-HDM to obtain an  extendible PDM$(m)$.
\end{proof}
\vskip 10pt

In view of Proposition \ref{pdmpcdp}, we employ a combination of construction
techniques for extendible PDM$(m)$ and $(3m,[3^m],3)$-PCDP.
The technique is recursive and so
we begin with some required small ingredients in the next subsection.

\subsection{Small Ingredients}

\begin{lemma}
\label{pdcp9}
There exists a $(3m,[3^m],3)$-PCDP for $m\in\{3,9\}$
\end{lemma}

\begin{proof}
 When $m=3$, take as base blocks $\{0,1,5\}$, $\{3,4,8\}$, and $\{6,7,2\}$.

When $m=9$, take as base blocks
$\{0,1,2\}$,
$\{3,4,6\}$,
$\{5,7,10\}$,
$\{8,11,16\}$,
$\{9,17,22\}$,
$\{12,18,24\}$,
$\{13,20,26\}$,
$\{14,21,25\}$, and
$\{15,19,23\}$.
\end{proof}
\vskip 10pt

\begin{lemma}
\label{pdm16}
There exists an extendible PDM$(m)$ for $m\in\{12,16,18,24,32,54\}$.
\end{lemma}

\begin{proof}
 An extendible PDM$(12)$ is listed below:
 \begin{equation*}
    \begin{bmatrix}
2&4&5&7& 8& 9& 10& 11&6\\
1&11&8&4&10&5&3 &9 &6\\
5& 1& 4 &11& 9& 2& 8& 7&6\\
\end{bmatrix}.
\end{equation*}
The three sets of elements missing from each row are $\{0,1,3\}$, $\{0,2,7\}$, and $\{0,3,10\}$.

For $m\in\{16,18,24,32,54\}$, 
we start with the $(m,4,1;2)$-HDM constructed in \cite{ChangMiao:2003}. 
First, remove the row of all zeros from each holey DM. Then
remove two columns as prescribed below: 

\begin{itemize}
\item for $m=16$: remove columns $[1,2,3]^{T}$ and $[-1,-2,-3]^{T}$.
\item for $m=18$: remove columns $[1,2,3]^{T}$ and $[-1,-2,-3]^{T}$.
\item for $m=24$: remove columns $[1,2,3]^{T}$ and $[-1,-2,-3]^{T}$.
\item for $m=32$: remove columns $[1,2,3]^{T}$ and $[-1,-2,-3]^{T}$.
\item for $m=54$: remove columns $[1,10,2]^{T}$ and $[2,12,5]^{T}$.
\end{itemize}
Finally, add the column $[m/2,m/2,m/2]^{T}$.  The 
resulting matrices have $m-3$ columns and the sets of missing elements
each row are:
\begin{itemize}
\item for $m=16$: $\{0,1,15\}$, $\{0,2,14\}$, and $\{0,3,13\}$.
\item for $m=18$: $\{0,1,17\}$, $\{0,2,16\}$, and $\{0,3,15\}$.
\item for $m=24$: $\{0,1,23\}$, $\{0,2,22\}$, and $\{0,3,21\}$.
\item for $m=32$: $\{0,1,31\}$, $\{0,2,30\}$, and $\{0,3,29\}$.
\item for $m=54$: $\{0,1,2\}$, $\{0,10,12\}$, and $\{0,2,5\}$.
\end{itemize}
\end{proof}
\vskip 10pt

\begin{lemma}
\label{pdm64}
There exists an extendible PDM$(m)$ for all
$m\in M$, where $M=\{36$, $48$, $64$, $72$, $96$, $108$, $128$, $144$, $162$, $192$, $256$, $288$, $384\}$.
\end{lemma}

\begin{proof}
For $m\in M\setminus\{36,108,288\}$, an
$(m,4,1;w)$-HMD exists with $w \in \{8,12,16,24\}$ \cite{ChangMiao:2003}.
Fill in the hole with an extendible PDM$(w)$ (which
exists by Example \ref{pdm8} or Lemma \ref{pdm16}) to obtain an extendible PDM$(m)$.

For $m=36$, take the $(m,4,1;2)$-HDM constructed in \cite{Yin:2005}, 
remove the two columns $[1,27,2]^{T}$, $[28,2,1]^{T}$, and add the column
$[18,18,18]^{T}$ to obtain a PDM$(36)$.

For $m\in\{108,288\}$, an $(m,4,1;w)$-HMD exists with 
$w \in \{12,24\}$ \cite{Yin:2005}.
Fill in the hole with an extendible PDM$(w)$ (which exists by Lemma \ref{pdm16})
to obtain an extendible PDM$(m)$.
\end{proof}

\subsection{Recursive Constructions}

\subsubsection{Recursive Constructions for Difference Matrices}
\hfill
\vskip 10pt
\begin{proposition}[Inflation, Yin \cite{Yin:2005}]
\label{product}
Suppose there exist an $(n,k,1;w)$-HDM and an $(m,k,1)$-CDM. Then
there exists a $(mn,k,1;mw)$-HDM.
\end{proposition}
\vskip 10pt

In Proposition \ref{product}, the $(n,k,1;w)$-HDM is said to be {\em inflated} by
the $(m,k,1))$-CDM to produce the $(mn,k,1;mw)$-HDM.

\vskip 10pt
\begin{theorem}[Chang and Miao \cite{ChangMiao:2003}]
\label{mult64}
If there exists an $(m,4,1;2)$-HDM, then there exists a $(64m,4,1;4m)$-HDM and a
$(72m,4,1;12m)$-HDM.
\end{theorem}
\vskip 10pt

\subsubsection{Recursive Constructions for PCDP}
\hfill

\vskip 10pt
\begin{proposition}
\label{product2}
Suppose there exist a $(3u,[3^{u}],3)$-PCDP and a homogeneous
 $(v,3,1)$-CDM.  Then 
there exists a $(3uv,3^{uv},3)$-PCDP.
\end{proposition}

\begin{proof}
 For each base block $\{a,b,c\}$ in the  $(3u,[3^{u}],3)$-PCDP, 
 we construct $v$ base blocks $\{a+3ud_0,b+3ud_1,c+3ud_2\}$, where 
 $[d_0,d_1,d_2]^T$ is a column of the homogeneous $(v,3,1)$-CDM.
It is easy to check that the resulting collection of base blocks is a PCDP.
\end{proof}
\vskip 10pt

\begin{proposition}
\label{6q}
Suppose there exists a $(6m,4,1;6)$-HDM. Then there exists an $(18m,[3^{6m}],3)$-PCDP.
\end{proposition}

\begin{proof}
Suppose there exists a $(6m,4,1;6)$-HDM, and hence a homogeneous
$(6m,3,1;6)$-HDM.  
For each column $[a,b,c]^{T}$ of the matrix, we construct a base block
$\{3a,3b+1,3c+2\}$. Then add the six base blocks $\{0,1,2\}$, $\{m,m+2,3m\}$, 
$\{m+1,3m+1,4m+1\}$, $\{2m,3m+2,5m+1\}$, $\{2m+1,4m+2,5m+2\}$, $\{2m+2,4m,5m\}$.
This results in an $(18m,[3^{6m}],3)$-PCDP.
\end{proof}
\vskip 10pt

\subsection{General Existence of Difference Matrices}

\begin{figure*}[!t]
\begin{equation}
\label{wide1}
B_1 = \left\{ 
\begin{bmatrix}
y_0 \\
cy_1 \\
(c+1)y_1
\end{bmatrix},  
\begin{bmatrix}
y_1 \\
-cy_1 \\
 -(c-1)y_0
\end{bmatrix}, 
\begin{bmatrix}
-y_0 \\
-cy_0 \\
-(c+1)y_0
\end{bmatrix}, 
\begin{bmatrix}
-y_1 \\
cy_0 \\
 (c-1)y_1
\end{bmatrix}
: y \in \bbZ_m^\boxtimes \right\}. 
\end{equation}
\begin{equation}
\label{wide2}
B_2 = \left\{ 
\pm \begin{bmatrix}
y_1 \\
cy_1 \\
 (1+c)y_0
\end{bmatrix}, 
\pm \begin{bmatrix}
cy_0 \\
y_1 \\
(1+c)y_1
\end{bmatrix},
\pm \begin{bmatrix}
c^2 y_0 \\
 cy_0 \\
 c(1+c)y_0 
 \end{bmatrix}, 
 \pm \begin{bmatrix}
 cy_1 \\
  c^2y_0 \\
   c(c+1)y_1
   \end{bmatrix}
: y \in \bbZ_m^\square \right\}.
\end{equation}
\hrulefill
\vspace*{4pt}
\end{figure*}

\begin{proposition}
\label{case2}
If $m >3 $ is prime and $m\equiv 3 \pmod{4}$, then there exists an extendible 
PDM$(2m)$.
\end{proposition}

\begin{proof}
 We employ the construction of Dinitz and Stinson \cite{DinitzStinson:1983}
 for a $(2m,4,1;2)$-HDM over $\bbZ_m \times \bbZ_2$.
Choose any $c \in \bbZ_{m}^\star$ such that $c^2 - 1 \in \bbZ_m^\boxtimes$
(this is where $m>3$ is required). Now let
$B_1$ be as defined in (\ref{wide1}). The $2(m-1)$ columns in $B_1$ form a 
homogeneous $(2m,3,1;2)$-HDM over  $\bbZ_m \times \bbZ_2$.
Remove the
columns $[1_0, c_0, (c+1)_0]^T$, $[4_0, 4c_0, 4(c+1)_0]^T$,
and add the column $[0_1,0_1,0_1]^T$.  
It is easy to check that this results in an  extendible 
PDM$(2m)$. 
The sets of elements missing from the first row to the last row are
$\{0_0,1_0,4_0\}$, $\{0_0, c_0, 4c_0\}$, and 
$\{0_0, (c+1)_0, 4(c+1)_0\}$, respectively.
Finally, we
note that $\gcd(m,2)=1$, and hence $\bbZ_{m} \times \bbZ_{2} \simeq \bbZ_{2m}$, which
is cyclic. 
\end{proof}
\vskip 10pt

\begin{proposition}
\label{case1}
If $m \equiv 1 \pmod{4}$ is a prime power, then there exists an extendible PDM$(2m)$.
\end{proposition}

\begin{proof}
We employ the construction of Dinitz and Stinson \cite{DinitzStinson:1983}
 for a $(2m,4,1;2)$-HDM over $\bbZ_m \times \bbZ_2$.
  Let $\omega$ be a primitive root in $\bbZ_{m}$ and let $c\in\bbZ_m^\boxtimes$. 
  Let $t = (m-1)/4$, and define $Q = \{\omega^0,\omega^2, \ldots, \omega^{2t-2}\}$.  Note that
\begin{equation*}
Q \cup (-Q) \cup cQ \cup (-cQ) = \bbZ_m \setminus \{0\}.
\end{equation*}
Now let $B_2$ as defined in (\ref{wide2}).
Remove from $B_2$
the columns $\pm[1_1,c_1,(1+c)_0]^T$ and add the column $[0_1,0_1,0_1]^T$.
The sets of elements missing from the first row to the last row are
$\{0_0, 1_1, -1_1\}$, $\{0_0, c_1, -c_1\}$, and 
$\{0_0, (c+1)_0, -(c+1)_0\}$, respectively.
Finally, we
note that $\gcd(m,2)=1$, and hence $\bbZ_{m} \times \bbZ_{2} \simeq \bbZ_{2m}$,
which is cyclic.
\end{proof}
\vskip 10pt

\begin{theorem}[Yin \cite{Yin:2005}]
\label{hdm}
 Let $m \geq 4$ be a product of the form
$2^{\alpha} 3^{\beta} p_1^{\alpha_1} \ldots p_t^{\alpha_t}$, where $p_j \geq 5$, $j\in[t]$.
Then there exists a $(2m,4,1;w)$-HDM if one of the following conditions is satisfied:
\begin{enumerate}[(i)]
\item $w =2$ and $(\alpha,\beta) \neq (1,0)$ or $(0,1)$;
\item $w=4$ and $(\alpha,\beta) = (1,0)$;
\item $w=6$ and $(\alpha,\beta) = (0,1)$.
\end{enumerate}
\end{theorem}

%
%

\subsection{Piecing Together}

The easier case when $m$ is odd is first addressed.

\subsubsection{The Case $m\equiv 1\pmod{2}$}
\hfill

\vskip 10pt
\begin{proposition}
\label{odd}
If $m$ is odd, then there exists a $(3m,[3^m],3)$-PCDP.
\end{proposition}

\begin{proof}
 If $(m,27)=1$, then the result follows from Theorem \ref{Con-1}.
If $(m,27)\in\{3,9\}$, then apply Proposition \ref{product2} 
with $u=\gcd(m,27)$ and $v=m/\gcd(m,27)$.
The existence of the ingredients is provided by Proposition \ref{CDM-01} and Lemma \ref{pdcp9}.
If $(m,27)=27$, then the result follows from Theorem \ref{Con-1} since
there exists a homogeneous $(m,3,1)$-CDM \cite{Yin:2005}.
\end{proof}
\vskip 10pt

\subsubsection{The Case $m\equiv 0\pmod{2}$}
\hfill

\vskip 10pt
\begin{proposition} 
\label{2,10mod12}
If $m\equiv 2,10 \pmod{12}$ and $m>2$, then there exists a $(3m,[3^m],3)$-PCDP.
\end{proposition}

\begin{proof}
 Write $m=2 p_1^{\alpha_1} \ldots p_t^{\alpha_t}$, where $p_j \geq 5$, $j\in[t]$.
 Inflate a $(m/p_1,4,1;2)$-HDM (which exists by Theorem \ref{hdm})
 by a $(p_1,4,1)$-CDM (which exists by Proposition \ref{CDM-01})
  to get an $(m,4,1,2p_1)$-HDM.  
Fill in the hole with an extendible PDM$(2p_1)$ 
 from Proposition \ref{case2} or Proposition \ref{case1} to obtain an extendible PDM$(m)$. 
 The result now follows from Proposition \ref{pdmpcdp}.
\end{proof}
\vskip 10pt

\begin{proposition} 
\label{4,8mod12}
Let $m \equiv 4,20 \pmod{24}$. Then there exists a  $(3m,[3^m],3)$-PCDP.
\end{proposition}

\begin{proof}
 By Theorem \ref{hdm}, there exists an  $(m,4,1;4)$-HDM, and therefore a homogeneous
$(m,3,1;4)$-HDM.  For each column $[a,b,c]^{T}$ of the matrix, we construct a base block 
$\{a_0,b_1,c_2\}$ on $\bbZ_{m} \times \bbZ_{3}$. 
Since $\gcd(m,3)=1$, $\bbZ_m\times\bbZ_3\simeq \bbZ_{3m}$.
Add four
blocks $(m/4)\{0,1,5\}$, $(m/4)\{2,6,9\}$,
$(m/4)\{3,4,10\}$,
and $(m/4)\{7,8,11\}$.  It is easy to check that it gives the desired result.
\end{proof}
\vskip 10pt

\begin{proposition}
\label{0mod6}
Let $m\equiv 0\pmod{6}$. Then there exists a $(3m,[3^m],3)$-PCDP.
\end{proposition}

\begin{proof}
Write $m=2^a3^bm'$, where $a,b\geq 1$ and $\gcd(m',6)=1$. We consider three cases:
\begin{description}
\item[$b=1$:]  When $m'=1$ and $a=1$, apply Proposition \ref{6q}
to a $(6,4,1;6)$-HDM (which exists trivially) to obtain an $(18,[3^6],3)$-PCDP.

When $m'=1$ and $2\leq a\leq 7$, 
 the result is obtained by applying Proposition \ref{pdmpcdp} to the extendible PDM$(m)$s
 obtained from Lemma
\ref{pdm16} and Lemma \ref{pdm64}.   

When $m'=1$ and $a\geq 8$,  
take a $(2^{a-6}\cdot 3,4,1;2)$-HDM, apply Theorem \ref{mult64} to obtain a 
$(2^a\cdot 3,4,1;2^{a-4}\cdot 3)$-HDM. Fill in the hole with
an extendible PDM$(2^{a-4}\cdot 3)$ (which exists by the induction hypothesis)
 to obtain an extendible PDM$(2^a\cdot 3)$.
Now apply Proposition \ref{pdmpcdp}.

When $m'>1$ and $a=1$, there exists a $(m,4,1;6)$-HDM by Theorem \ref{hdm}.
Now apply Proposition \ref{6q}.

When $m'>1$ and $a \geq 2$, let $p$ be a prime factor of $m'$ (note that $p\geq 5$).
Theorem \ref{hdm} implies the existence
of an $(m/p,4,1;2)$-HDM. Inflate this $(m/p,4,1;2)$-HDM by
 a $(p,4,1)$-CDM (which exists by
Proposition \ref{CDM-01})
to obtain an $(m,4,1;2p)$-HDM. Fill in the hole with an extendible
PDM$(2p)$ from Proposition \ref{case2} or Proposition \ref{case1} to obtain an
extendible PDM$(m)$. Now apply Proposition \ref{pdmpcdp}.

\item[$b=2$:] When $m'=1$ and $a\in[5]$, an extendible PDM$(m)$ exists by Lemma \ref{pdm16}
and Lemma \ref{pdm64}.

When $m'=1$ and $a=6$, apply Theorem \ref{mult64} with an 
$(8,4,1;2)$-HDM (which exists by Theorem \ref{hdm})
to obtain a $(576,4,1;96)$-HDM. Fill in the hole with an extendible PDM$(96)$
from Lemma \ref{pdm64} to obtain an extendible PDM$(576)$
and apply Proposition \ref{pdmpcdp}.

When $m'=1$ and $a\geq 7$, apply Theorem \ref{mult64} to
a $(m/64,4,1;2)$-HDM (which exists by Theorem \ref{hdm}) to obtain a
$(m,4,1;m/16)$-HDM. Fill in the hole with an extendible PDM$(m/16)$ (which
exists by the induction hypothesis) to obtain an extendible PDM$(m)$. Now
apply Proposition \ref{pdmpcdp}.

When $m'>1$, let $p$ be a prime factor of $m'$ (note that $p\geq 5$). Theorem \ref{hdm} implies
the existence of a $(m/p,4,1;2)$-HDM. Inflate this holey DM by a $(p,4,1)$-CDM
(which exists by Proposition \ref{CDM-01}) to obtain a $(m,4,1;2p)$-HDM.
Fill in the hole with an extendible PDM$(2p)$ (which exists by Proposition \ref{case2}
or Proposition \ref{case1}) to obtain an extendible PDM$(m)$. Now apply
Proposition \ref{pdmpcdp}.

\item[$b\geq 3$:] Theorem \ref{hdm} implies the existence of an $(m/27,4,1;w)$-HDM,
for some $w\in\{2,4,6\}$.  Inflate this $(m/27,4,1;w)$-HDM by a $(27,4,1)$-CDM 
(which exists by Proposition \ref{CDM-2}) to obtain an
$(m,4,1;27w)$-HDM. Fill in the hole 
with an extendible PDM$(27w)$ (which exists by
Example \ref{pdm8} or Lemma \ref{pdm64}) to obtain
an extendible PDM$(m)$. The result now follows from Proposition \ref{pdmpcdp}.
\end{description}
\end{proof}
\vskip 10pt

\subsubsection{Summary}
\hfill

\vskip 10pt
\begin{theorem}
There exists a $(3m,[3^m],3)$-PCDP for all $m\not\equiv 8,16\pmod{24}$, except
when $m=2$.
\end{theorem}

\begin{proof}
Propositions \ref{odd}, \ref{2,10mod12}, \ref{4,8mod12}, and \ref{0mod6} give a
$(3m,[3^m],3)$-PCDP for all $m\not\equiv 8,16\pmod{24}$, $m>2$.
It is easy to check that a
$(6,[3^2],3)$-PCDP
cannot exist.
\end{proof}

\section{Concluding Remarks}

 In this paper, a number of new infinite families of optimal  
 PCDPs are presented. The PCDPs obtained
 can be used directly to produce
 frequency-hopping
 sequences optimal with respect to Hamming correlation
 and comma-free codes optimal with respect to redundancy. 
 They are also of independent interest in
 combinatorial design theory.
 

\begin{thebibliography}{10}
\providecommand{\url}[1]{#1}
\csname url@rmstyle\endcsname
\providecommand{\newblock}{\relax}
\providecommand{\bibinfo}[2]{#2}
\providecommand\BIBentrySTDinterwordspacing{\spaceskip=0pt\relax}
\providecommand\BIBentryALTinterwordstretchfactor{4}
\providecommand\BIBentryALTinterwordspacing{\spaceskip=\fontdimen2\font plus
\BIBentryALTinterwordstretchfactor\fontdimen3\font minus
  \fontdimen4\font\relax}
\providecommand\BIBforeignlanguage[2]{{%
\expandafter\ifx\csname l@#1\endcsname\relax
\typeout{** WARNING: IEEEtran.bst: No hyphenation pattern has been}%
\typeout{** loaded for the language `#1'. Using the pattern for}%
\typeout{** the default language instead.}%
\else
\language=\csname l@#1\endcsname
\fi
#2}}

\bibitem{Pickholtzetal:1982}
R.~L. Pickholtz, D.~L. Schilling, and L.~B. Milstein, ``Theory of spread
  spectrum communications -- a tutorial,'' \emph{IEEE Trans. Commun.}, vol.~30,
  no.~5, pp. 855--884, May 1982.

\bibitem{Fuji-Haraetal:2004}
R.~Fuji-Hara, Y.~Miao, and M.~Mishima, ``Optimal frequency hopping sequences: a
  combinatorial approach,'' \emph{IEEE Trans. Inform. Theory}, vol.~50, pp.
  2408--2420, 2004.

\bibitem{Levenshtein:1971b}
V.~I. Levenshte{\u\i}n, ``One method of constructing quasi codes providing
  synchronization in the presence of errors,'' \emph{Problems of Information
  Transmission}, vol.~7, no.~3, pp. 215--222, 1971.

\bibitem{Geetal:2006}
G.~Ge, R.~Fuji-Hara, and Y.~Miao, ``Further combinatorial constructions for
  optimal frequency-hopping sequences,'' \emph{J. Comb. Theory Ser. A}, vol.
  113, pp. 1699--1718, 2006.

\bibitem{Golombetal:1958}
S.~W. Golomb, B.~Gordon, and L.~R. Welch, ``Comma-free codes,'' \emph{Canad. J.
  Math.}, vol.~10, no.~2, pp. 202--209, 1958.

\bibitem{Tonchev:2005}
V.~D. Tonchev, ``Partitions of difference sets and code syncrhonization,''
  \emph{Finite Fields Appl.}, vol.~11, pp. 601--621, 2005.

\bibitem{Levenshtein:2004}
V.~I. Levenshte{\u\i}n, ``Combinatorial problems motivated by comma-free
  codes,'' \emph{J. Combin. Des.}, vol.~12, pp. 184--196, 2004.

\bibitem{Jungnickeletal:2007}
D.~Jungnickel, A.~Pott, and K.~W. Smith, ``Difference sets,'' in \emph{The
  {CRC} Handbook of Combinatorial Designs}, C.~J. Colbourn and J.~H. Dinitz,
  Eds.\hskip 1em plus 0.5em minus 0.4em\relax Boca Raton, FL: Chapman \& Hall,
  2007, pp. 419--435.

\bibitem{Lempeletal:1977}
A.~Lempel, M.~Cohn, and W.~L. Eastman, ``A class of binary sequences with
  optimal autocorrelation properties,'' \emph{IEEE Trans. Inform. Theory},
  vol.~23, pp. 38--42, 1977.

\bibitem{Arasuetal:2001}
K.~T. Arasu, C.~Ding, T.~Helleseth, P.~Vijay~Kumar, and H.~Martinsen, ``Almost
  difference sets and their sequences with optimal autocorrelation,''
  \emph{IEEE Trans. Inform. Theory}, vol.~47, pp. 2834--2843, 2001.

\bibitem{Dingetal:2001}
C.~Ding, T.~Helleseth, and H.~Martinsen, ``New families of binary sequences
  with optimal three-level autocorrelation,'' \emph{IEEE Trans. Inform.
  Theory}, vol.~47, pp. 428--433, 2001.

\bibitem{Ge:2005}
G.~Ge, ``On $(g,4;1)$-difference matrices,'' \emph{Discrete Math.}, vol. 301,
  pp. 164--174, 2005.

\bibitem{Yin:2005}
J.~Yin, ``Cyclic difference packing and covering arrays,'' \emph{Des. Codes
  Cryptogr.}, vol.~37, no.~2, pp. 281--292, 2005.

\bibitem{ChangMiao:2003}
Y.~Chang and Y.~Miao, ``Constructions for optimal optical orthogonal codes,''
  \emph{Discrete Math.}, vol. 261, pp. 127--139, 2003.

\bibitem{DinitzStinson:1983}
J.~H. Dinitz and D.~R. Stinson, ``{MOLS} with holes,'' \emph{Discrete Math.},
  vol.~44, pp. 145--154, 1983.

\end{thebibliography}

\section*{Authors' Biographies}

{\bf Yeow Meng Chee} (SM'08) received the B.Math. degree in computer science and
combinatorics and optimization and the M.Math. and Ph.D. degrees in computer science,
from the University of Waterloo, Waterloo, ON, Canada, in 1988, 1989, and 1996, respectively.

Currently, he is an Associate Professor at 
the Division of Mathematical Sciences, School of Physical
and Mathematical Sciences, Nanyang Technological University, Singapore.
Prior to this, he was Program Director of Interactive Digital Media R\&D in the
Media Development Authority of Singapore,
Postdoctoral Fellow at the University of Waterloo and
IBM's Z{\"u}rich Research Laboratory, General Manager of the Singapore Computer Emergency
Response Team, and 
Deputy Director of Strategic Programs at the Infocomm Development Authority, Singapore.
His research interest lies in the interplay between combinatorics and computer science/engineering,
particularly combinatorial design theory, coding theory, extremal set systems,
and electronic design automation.

\vskip 10pt 

{\bf Alan C. H. Ling} was born in Hong Kong in 1973. He received the B.Math.,
M.Math., and Ph.D. degrees in combinatorics \& optimization from the University
of Waterloo, Waterloo, ON, Canada, in 1994, 1995, and 1996, respectively.
He worked at the Bank of Montreal, Montreal, QC, Canada, and Michigan
Technological University, Houghton, prior to his present position as Associate
Professor of Computer Science at the University of Vermont, Burlington. His research
interests concern combinatorial designs, codes, and applications in computer
science.

\vskip 10pt

{\bf Jianxing Yin} received the B.Sc. degree from Suzhou University, China, in 1977.
Since 1977, he has been a teacher in the Department of Mathematics, Suzhou University, China.
He is currently a Full Professor at the same university. His research interests include
applications of combinatorial designs in coding theory and cryptography.

\end{document}